\documentclass[11pt,a4paper]{article}

\usepackage[T1]{fontenc}
\usepackage[utf8]{inputenc}
\usepackage{lmodern}                %
\usepackage[margin=1in]{geometry}   %
\usepackage{amsmath,amssymb,amsthm,mathtools}
\newtheorem{theorem}{Theorem}[section]
\newtheorem{lemma}[theorem]{Lemma}
\newtheorem{observation}[theorem]{Observation}

\newtheorem{corollary}[theorem]{Corollary}
\newtheorem{conjecture}[theorem]{Conjecture}
\usepackage{hyperref}
\usepackage{doi}
\usepackage{cleveref}
\usepackage{xcolor}
\theoremstyle{definition}
\newtheorem{definition}[theorem]{Definition}
\newtheorem{remark}[theorem]{Remark}

\title{CNFs and DNFs with Exactly $k$ Solutions}

\author{%
  L.\,Sunil Chandran\thanks{Indian Institute of Science, Bengaluru, India.
  \href{mailto:sunil@iisc.ac.in}{sunil@iisc.ac.in}.}
  \and
  Rishikesh Gajjala\thanks{Indian Institute of Science, Bengaluru, India.
  \href{mailto:rishikeshg@iisc.ac.in}{rishikeshg@iisc.ac.in}.}
  \and
  Kuldeep S.\ Meel\thanks{Georgia Institute of Technology, USA.
  \href{mailto:meel@gatech.edu}{meel@gatech.edu}.}
}

\date{ }

\begin{document}
\maketitle

\begin{abstract}
	Model counting is a fundamental problem that consists of determining the 
	number of satisfying assignments for a given Boolean formula. The weighted variant, which computes the weighted sum of satisfying 
	assignments, has extensive applications in probabilistic reasoning, 
	network reliability, statistical physics, and formal verification. A 
	common approach for solving weighted model counting is to reduce it to 
	unweighted model counting, which raises an important question: {\em What is 
	the minimum number of terms (or clauses) required to construct a DNF 
	(or CNF) formula with exactly $k$ satisfying assignments?}
	
	In this paper, we establish both upper and lower bounds on this 
	question. We prove that for any natural number $k$, one can construct a 
	monotone DNF formula with exactly $k$ satisfying assignments using at 
	most $O(\sqrt{\log k}\log\log k)$ terms. This construction represents 
	the first $o(\log k)$ upper bound for this problem. We complement this 
	result by showing that there exist infinitely many values of $k$ for 
	which any DNF or CNF representation requires at least $\Omega(\log\log 
	k)$ terms or clauses. These results have significant implications for 
	the efficiency of model counting algorithms based on formula 
	transformations.
\end{abstract}
\section{Introduction}
\label{sec:introduction}

The (unweighted) model counting is a classical problem in which one has 
to find the number of satisfying assignments for a given boolean 
formula. Usually, Boolean formulae are considered in two forms---Conjunctive Normal Form (CNF) and Disjunctive Normal Form (DNF). In 
the former case, the formula is written as conjunctions (ANDs) of 
clauses. (Clauses are literals combined using ORs.) In the latter case 
the formula is written as disjunctions (ORs) of terms. (Terms are 
literals combined using ANDs). Note that a variable in a formula may 
appear either in positive form or in negated form. A monotone formula 
in DNF (or in CNF) consists of only variables in positive form.

The unweighted model counting problem was shown to be \#P-complete for 
formulae in both CNF and DNF by Valiant \cite{Valiant79}. The weighted 
model counting is a generalization of this problem. We are not 
providing a complete definition of this problem here since it is 
technical, but it can be found in \cite{Meel15}. Weighted model 
counting has been extensively studied \cite{ DBLP:journals/jair/DarwicheM02, DBLP:journals/siamcomp/FlumG04} due to its diverse applications 
across multiple domains. These applications include probabilistic 
reasoning \cite{Roth93}, network reliability estimation 
\cite{SangBBKP04}, statistical physics, probabilistic databases 
\cite{olteanu_probabilistic_2011}, program synthesis, and system verification 
\cite{DomshlakH07, XueCD12}. The fundamental nature of weighted model 
counting has led to its emergence as a core computational problem in 
areas requiring reasoning under uncertainty, where the ability to 
compute weighted sums across large combinatorial spaces is essential.

A natural approach for solving the weighted model counting problem is 
to reduce it to the unweighted model counting problem and then to use 
the existing solvers for the latter \cite{Meel15}. This reduction 
approach has proven effective across various application domains as it 
leverages the significant advances in unweighted model counting 
algorithms. A key subroutine in this reduction involves finding DNFs or 
CNFs with exactly $k$ satisfying assignments for a given positive 
integer $k$. The reduction becomes more efficient as the number of 
terms (or clauses) in the DNF or CNF decreases. {While the number of terms is not the only factor of relevance, it raises a very natural question: {\em What is 
	the minimum number of terms (or clauses) required to construct a DNF 
	(or CNF) formula with exactly $k$ satisfying assignments?}}

This can be quantified using $\beta(k)$, which is defined as follows:

\begin{definition}
	The minimum number of terms or clauses needed to generate a DNF or CNF 
	with exactly $k$ satisfying assignments is defined to be $\beta(k)$.
\end{definition}

{It is known that $\beta(k) = O(\log{k})$~\cite{Meel15}}. In this work, using an interesting connection to ideals of set 
systems, we give new lower and upper bounds on $\beta(k)$.

It is easy to see that many arbitrarily large numbers $k$ exist for 
which $\beta(k)=1$. For example, consider the DNF in which $(x_1)$ is 
the only term and $x_2,x_3\cdots x_{q+1}$ are the free variables, i.e., 
the variables which do not appear in the DNF (and hence their truth 
value does not affect the satisfiability of the formulae). Therefore, 
all $k$ of the form $2^q$ can be generated using a DNF with only $1$ 
term, i.e., $\beta(k)=1$. At the same time, there exists $k$ for which 
we need at least $\Omega(\log\log{k})$ terms to generate a DNF with 
exactly $k$ solutions (we prove this statement later). Thus, $\beta(k)$ 
does not increase or decrease monotonically with $k$. This motivates us 
to introduce a parameter called \textit{block count} of $k$, which is 
more intimately associated with the number of terms needed to generate 
a DNF with exactly a given number of satisfying assignments.

Let $\mathbf{1}_m$ denote $\mathbf{1}\mathbf{1}\cdots m \text{ times}$ 
and $\mathbf{0}_q$ denote $\mathbf{0}\mathbf{0}\cdots q \text{ times}$. 
Using this notation, the binary representation of $49$, namely $110001$ 
can be represented as $\mathbf{1}_2\mathbf{0}_3\mathbf{1}_1$.

\begin{definition}
	The block binary representation of any $k\in \mathbb{N}$, is defined to 
	be the unique representation $\mathbf{1}_{q_b} \mathbf{0}_{l_b}\cdots 
	\mathbf{1}_{q_2}\mathbf{0}_{l_2} \mathbf{1}_{q_1} \mathbf{0}_{l_1}$ 
	where $q_i> 0$ and $l_j> 0$ for all $i\in [b]$ and $j \in [2,b]$. Note 
	that $l_1$ can be $0$. For $k$ with such a representation, its block 
	count, $bl(k)=b$.
\end{definition}

The main result of our paper establishes a relationship between the 
block count of a number and the minimum number of terms needed to 
construct a DNF with exactly that many satisfying assignments.

\begin{theorem}[Main Result]\label{thm:sqrt_bnd}
	For every $k\geq 3$,  $$\log{(bl(k)+1)} \leq  \beta(k) \leq \min\left\{ 20\sqrt{\log{k}} \log\log{k},\ {bl}(k) + 1 \right\}$$
\end{theorem}

This represents the first $o(\log{k})$ construction for this problem. 
We also conjecture that the value of $\beta(k)$ is polynomial in 
$\log{(bl(k)+1)}$.

\begin{conjecture}
	There exists a sufficiently large constant $C$ and a function $f(x)$, 
	which is polynomial in $x$ such that for every $k\in \mathbb{N}$, 
	$\beta(k) \leq C \cdot f(\log{(bl(k))})$.
\end{conjecture}

\paragraph*{Organization:} The rest of the paper is organized as follows. In Section~\ref{sec:ideals}, we 
establish a connection between the problem of finding minimum-sized 
DNFs with exactly $k$ satisfying assignments and the theory of ideals 
of set systems. Section~\ref{sec:lower} provides the proof of our lower bound, while  Section ~\ref{sec:upper}, which forms the technical core of 
the paper, presents the proof of upper bound. We finally conclude in Section~\ref{sec:conclusion}.

\section{Connection to Ideals of Sets System}\label{sec:ideals}
We now show that a natural problem on the ideals of sets 
system (which is also of independent interest) is equivalent 
to finding small monotone DNFs (formulas consisting of 
only variables in positive form) with exactly a given number 
of satisfying assignments. We use this formulation to 
derive our upper bounds for $\beta(k)$.

\subsection{Notation}
$\mathbb{N}$ denotes the set of natural numbers. We use 
$\log{k}$ to denote $\log_{2}{k}$. For a set $S$, $|{S}|$ 
and $2^{{S}}$ denote its cardinality and power set, 
respectively. The union of $A$ and $B$ is denoted as $A\cup B$. The union of two disjoint sets $A,B$ is denoted as $A \sqcup B$.  The notation $[a,b]$ represents 
$\{a,a+1\cdots b\}$ and $[b]$ represents $[1,b]$. 
For every number $i$, we create distinct copies 
$i_0,i_1,i_2\cdots$. The set $[w]_i$ represents 
$\{1_i,2_i\cdots w_i\}$. Note that the sets 
$[w],[w]_0,[w]_1\cdots$ are all different from each 
other as these sets are pairwise disjoint. Given a family of 
sets $\mathcal{S}=\{S_1,S_2\cdots S_t\}$ and a set $X$, we 
define $\mathcal{S}+X= \{S_1\cup X,S_2\cup X\cdots S_t\cup X\}$.

An anti-chain is a subset $\mathcal{A}$ of a partially 
ordered set $P$ such that any two distinct elements of 
$\mathcal{A}$ are incomparable. An (order) ideal (also 
called semi-ideal, down-set, or monotone decreasing subset) 
of $P$ is a subset $I$ of $P$ such that if $t \in I$ and 
$s \leq t$, then $s \in I$. Similarly, a dual order ideal 
(also called up-set or monotone increasing subset) is a 
subset $I$ of $P$ such that if $t \in I$ and $s \geq t$, 
then $s \in I$ \cite{stanley}. When $P$ is finite, there is 
a one-to-one correspondence between anti-chains of $P$ and 
order ideals: the anti-chain $\mathcal{A}$ associated with 
the order ideal $I$ is the set of maximal elements of $I$, 
while $I = \{\, s \in P \mid s \leq t \text{ for some } 
t \in \mathcal{A} \,\}$. Then the anti-chain $\mathcal{A}$ 
is said to generate the ideal $I=\mathbf{ID}(\mathcal{A})$.

\begin{remark}
	There may be some difference of opinion with the definition 
	of ideal given above since, in some contexts, a slightly 
	different definition is used for ideals. However, in this 
	paper, we only study set systems, and with respect to set 
	systems, most authors use the above definition for ideals. 
	For example, see Bollab\'as \cite{bollobasbook}.
\end{remark}

\subsection{Problem Definition}
\begin{definition}
	The ideal generated by a family of sets, 
	$\mathcal{S}=\{S_1, S_2\cdots S_\alpha \}$, is 
	$\mathbf{ID}(\mathcal{S}) = 2^{S_1} \bigcup 2^{S_2} 
	\bigcup \cdots 2^{S_\alpha}$. 
\end{definition}
Note that the minimal family 
	of sets that generates a given ideal is an antichain. 
\begin{definition}
	Given a natural number $k$, $\alpha(k)$ is defined to be 
	the minimum $|\mathcal{S}|$ for which 
	$|\mathbf{ID}(\mathcal{S})|=k$.
\end{definition}

Observe that for every natural number $k$, there is a family 
of sets $\mathcal{S}=\{\emptyset, \{1\},\{2\}\cdots \{k-1\}\}$ 
such that $|\mathbf{ID}(\mathcal{S})|=k$. Therefore, 
$\alpha(k)$ exists for all $k$ and moreover, $\alpha(k) \leq k$. 
In this work, we establish more meaningful bounds on $\alpha(k)$.

\subsection{Combinatorial Background}
Ideals and their symmetric counterpart filters are central 
concepts in the study of set systems. These concepts appear 
in some of the most fundamental theorems regarding set systems. 
For example, Bollobás and Thomason proved that every non-trivial 
monotone increasing/decreasing property of subsets of a set has 
a threshold function \cite{ThresholdF}, in the probabilistic 
model where each element is chosen with probability $p$. Here, 
monotone decreasing property corresponds to ideals. This is one 
of the most significant results in the theory of random graphs 
(see chapter $6$ of \cite{bollobasbook}).

Another well-known result on ideals and dual order ideals is 
Kleitman's lemma \cite{kleitman}, which triggered a long line 
of research on correlation-type inequalities, culminating in 
the Four Functions Theorem of Ahlswede and Daykin 
\cite{fourfunction} (see chapter $6$ of \cite{alonspencer}). 
When studying extremal problems on set systems, it is often 
sufficient to prove the extremality restricted to set systems 
that are ideals or dual-order ideals. For example, see 
Kleitman's proof establishing a tight bound for the cardinality 
of maximal $l$-intersecting families \cite{KLEITMAN1966209} 
(see chapter 13 of \cite{bollobasbook}).

In chapter $17$ of \cite{bollobasbook}, Bollab\'as discusses 
theorems of the form $(m,k) \rightarrow (r,s)$ regarding traces. 
Such a theorem means if the universe $X = [k]$ and a family 
$\mathcal{F}$ consists of $m$ subsets of $X$, then there exists 
an $s$-element subset $S$ of $X$ such that when we take the 
intersection of $S$ with the members of $\mathcal{F}$, we get 
at least $r$ distinct subsets. Alon \cite{ALON1983199} and 
Frankl \cite{Frankl1983OnTT} independently proved that to 
establish any theorem of the form $(m,k) \rightarrow (r,s)$, 
it is sufficient to prove the corresponding statement when 
$\mathcal{F}$ is restricted to an ideal \cite{Frankl1983OnTT}.

Ideals are also studied under the name abstract simplicial 
complex or abstract complex. This represents a combinatorial 
description of the geometric notion of a simplicial complex 
\cite{Lee_2011}. In the context of matroids and greedoids, 
these are also referred to as independence systems 
\cite{Korte_Schrader_Lovasz_1991}.

Several questions closely related to our work have been studied 
in the literature. For instance, Duffus, Howard, and Leader 
investigated the maximum cardinality of an anti-chain that can 
be present in a given ideal \cite{DUFFUS201946}. The problem 
we discuss in this paper—finding the minimum cardinality 
$\alpha(k)$ of the anti-chain that can generate an ideal of 
a given size $k$—has been examined by both computer scientists 
and combinatorialists due to its applications in model counting.

\subsection{Connection between Ideals and Monotone DNFs} 
The model counting problem for monotone DNF formulae has an 
interesting connection to the ideals of set systems. In fact, 
these problems are essentially equivalent. Let $x_1,x_2,\ldots,x_m$ 
be the set of positive literals used in the monotone DNF 
formulae we consider. An assignment assigns a truth value to 
each of these $m$ variables; if the formula evaluates to TRUE 
under this assignment, then the assignment is called a 
satisfying assignment.

Take the universe $U = [m] = \{1,\ldots,m\}$. Given any subset 
$S$ of $U$, we can associate to $S$ a term 
$T_S= \bigwedge_{i \in S} x_i$. Conversely, given a term 
$T= \{ x_{i_1} \land x_{i_2} \land \ldots \land x_{i_t} \}$ 
of a monotone DNF formula, we can associate the subset 
$S_T = \{i_1, i_2, \ldots, i_t\}$ to $T$. Also, to a family 
$\mathcal{F} = \{ S_1, S_2, \ldots, S_t\}$ of subsets of $U$, 
we can associate a monotone DNF formula 
$f_{\mathcal{F}} = T(S_1) \lor T(S_2) \lor \ldots \lor T(S_t)$. 
Conversely, to a monotone DNF formula 
$f = T_1 \lor T_2 \lor \ldots \lor T_{\ell}$, we can associate 
a family of subsets of $U$, namely 
$\mathcal{F}_f = \{S_{T_1}, S_{T_2}, \ldots, S_{T_\ell} \}$.

Thus, there is a one-to-one correspondence between monotone 
DNF formulae using variables $x_1,x_2,\ldots,x_m$ and families 
of subsets of $U$.

Let $\mathcal{F}$ be a family of subsets. Then let 
$\overline{\mathcal{F}} = \{\overline{S}: S \in \mathcal{F}\}$, 
where $\overline{S} = U \setminus S$ is the complement of $S$. 
We can show that the set of satisfying assignments of a 
monotone DNF formula $f$ has a one-to-one correspondence with 
the ideal generated by the family $\overline{\mathcal{F}_f}$. 
This is because for $f$ to be satisfied, at least one term of 
$f$ must be satisfied. If term $T_i$ is satisfied, then all 
literals appearing in $T_i$ must be set to TRUE. So, the set 
of literals that are set to FALSE must correspond to a set of 
indices $S'$ such that $S' \subseteq \overline{S_{T_i}}$. In 
other words, $S' \in \mathbf{ID}(\overline{\mathcal{F}_f})$.

The converse is also true: For 
$S' \in \mathbf{ID}(\overline{\mathcal{F}_f})$, the assignment 
where each variable $x_i$ with $i \in S'$ is set to FALSE and 
the remaining variables set to TRUE will be a satisfying 
assignment for $f$. This is because there will be a superset 
of $S'$ in $\overline{\mathcal{F}_f}$, and the term in $f$ 
that corresponds to the complement of this superset would 
evaluate to TRUE. Since every satisfying assignment can be 
bijectively mapped to the set of variables that are set to 
FALSE, the set of satisfying assignments of $f$ are 
bijectively mapped to the ideal of 
$\overline{\mathcal{F}_f}$. From this discussion, we have:

\begin{theorem}\label{thm:ideals_dnf}
	Let $k$ be a positive integer. If $\mathcal{F}$ is a family 
	of subsets with $|\mathbf{ID}(\mathcal{F})| = k$, then there 
	exists a monotone DNF formula 
	$f = f_{\overline{\mathcal{F}_f}}$ with exactly $k$ 
	satisfying assignments. In particular, a monotone DNF formula 
	with the smallest number of terms and exactly $k$ satisfying 
	assignments will have $\alpha(k)$ terms.
\end{theorem}
\begin{corollary}\label{alpha_beta_simp_conn}
For $k\geq 1$, $\alpha(k) \geq \beta(k)$    
\end{corollary}

\begin{remark}
	A similar statement can be made about monotone CNF formulae. 
	The difference is that $k$ would represent the number of 
	non-satisfying assignments, and a subset in the ideal would 
	correspond to the variables assigned TRUE.
\end{remark}

\section{Proof of Lower Bound}\label{sec:lower}
We first state the inclusion-exclusion principle
\begin{theorem}\label{inclexclthmpwset}
	For finite sets $V_1,V_2\cdots V_q$
	$$
	|\bigcup_{i=1}^{q} V_i | = \sum_{\emptyset \neq J \subseteq \{1,2\cdots,q\}}(-1)^{|J|+1}  \Biggl|\bigcap_{j \in J} V_j\Biggr|
	$$    
\end{theorem}
For a formulae $\mathcal{F}$, let $Sol(\mathcal{F})$ denote the set of satisfying assignments for $\mathcal{F}$.

\begin{observation}\label{dnfinclexclthmpwset}
	For a  DNF formula $\mathcal{F} =T_1 \lor T_2 \cdots \lor T_{q}$, we have $$|Sol(\mathcal{F})|=\sum_{\emptyset \neq J \subseteq \{1,2\cdots,q\}}(-1)^{|J|+1}  \Biggl| Sol(\bigwedge_{j \in J} T_j)\Biggr|$$
\end{observation}

\begin{proof}
	For the DNF formula $\mathcal{F} =T_1 \lor T_2 \cdots \lor T_{q}$, it is easy to see that $$Sol(\mathcal{F}) = \bigcup_{i=1}^{q} Sol(T_i)$$. Therefore, from \cref{inclexclthmpwset}
	$$
	| \bigcup_{i=1}^{q} Sol(T_i) | = \sum_{\emptyset \neq J \subseteq \{1,2\cdots,q\}}(-1)^{|J|+1}  \Biggl|\bigcap_{j \in J} Sol(T_j)\Biggr|
	$$ 
	Observe that $Sol(T_i) \cap Sol(T_j) = Sol(T_i \land T_j)$. Therefore, 
	$$\sum_{\emptyset \neq J \subseteq \{1,2\cdots,q\}}(-1)^{|J|+1}  \Biggl|\bigcap_{j \in J} Sol(T_j)\Biggr| = \sum_{\emptyset \neq J \subseteq \{1,2\cdots,q\}}(-1)^{|J|+1}  \Biggl| Sol(\bigwedge_{j \in J} T_j)\Biggr|
	$$
	
\end{proof}

\begin{observation}\label{2powerobs}
	For any non-empty set $J \subseteq \{1,2\cdots,q\}$, the value of $\Biggl| Sol(\bigwedge_{j \in J} T_j)\Biggr|$ is either $0$ or of the form $2^{\alpha}$ for some $\alpha \in \mathbb{N} \cup \{0\}$.
\end{observation}
\begin{proof}
	Consider a literal $y$. We have the following cases
	\begin{itemize}
		\item If the literal $y$ appears positively in $T_i$ and negatively in $T_j$ for some $i,j \in J$, then $\bigwedge_{j \in J} T_j$ has no solution.
		\item If the literal $y$ appears positively for at least one $T_i$ for $i \in J$ and never appears negatively for any $T_j$ for $j \in J$, then $y$ must be true in all satisfying assignments.
		\item  Similarly, if the literal $y$ appears negatively for at least one $T_i$ for $i \in J$ and never appears positively for any $T_j$ for $j \in J$, then $y$ must be false in all satisfying assignments.
		\item If the literal $y$ does not appear in any $T_j$ for $j \in J$, then it can take the true or false value in an assignment. 
	\end{itemize}
	It is now easy to see that if there are $\alpha$ such literals which never appeared in any $T_j$ for $j \in J$, there will be $2^{\alpha}$ many satisfying assignments. 
\end{proof}

\begin{observation}\label{obs:num_blocks}
	Let $t \in \mathbb{N}$ and $x_i,y_i\in \{0\}\cup \mathbb{N}$ for all $i\in [t]$. If $k=\sum\limits_{i\in [t]}(-1)^{x_i}2^{y_i} \geq 1$, then 
	$bl(k)\leq t$.  
\end{observation}
\begin{proof}
We prove this by induction on $t$. When $t=1$, $x_1$ must be even since $k\geq 1$. Therefore, $k$ is of the form $2^{y_1}$ and hence $bl(k)=1=t$.
	
Let $t>1$ and $m=k-\min\limits_{i\in [t]} (-1)^{x_i}2^{y_i}$. For $i\in [t]$, if there exists an $x_i$ which is odd, then $\min\limits_{i\in [t]} (-1)^{x_i}2^{y_i} <0$ and hence $m \geq 1$. On the other hand, if for all $i\in [t]$, $x_i$ is even, then $m$ can be written as the sum of positive integers. Therefore, $m\geq 1$. Therefore, by the induction assumption $bl(m)\leq t-1$.

Observe that when $(-1)^{x_i}2^{y_i}$ is added to the binary representation of $m$, $1$ is either added or subtracted at the $(y_i+1)$th bit of $m$. Suppose the $(y_i+1)$th bit of $m$ be $0$ (respectively $1$). Then adding (respectively subtracting) $1$ at the $(y_i+1)$th position changes the number of blocks by at most $1$. On the other hand, when the $(y_i+1)$th bit of $m$ is $1$ (respectively $0$), adding (respectively subtracting) $1$ at the $(y_i+1)$th position flips all contiguous $1$s (respectively $0$s) at and before $(y_i+1)$th position and the first preceding $0$ (respectively $1$). Therefore, the number of blocks change by at most $1$.

Therefore, as $k=m+(-1)^{x_t}2^{y_t}$, $bl(k)\leq bl(m)+1 \leq t$.
\end{proof}

\begin{lemma}
	For every $k\in \mathbb{N}$, $\log (bl(k)+1) \leq \beta(k)$    
\end{lemma}

\begin{proof}
	Towards a contradiction, let there exist some $k\in \mathbb{N}$ such that $\log (bl(k)+1) > \beta(k)$. Let $\mathcal{F}= T_1 \lor T_2 \cdots \lor T_{\beta(k)} $ be a DNF such that it has exactly $k$ solutions. From \cref{dnfinclexclthmpwset},
	
	$$|Sol(\mathcal{F})|=\sum_{\emptyset \neq J \subseteq \{1,2\cdots,\beta(k)\}}(-1)^{|J|+1}  \Biggl| Sol(\bigwedge_{j \in J} T_j)\Biggr|$$

	From \cref{2powerobs}, it is easy to see that $|Sol(\mathcal{F})|$ can be written as the sum or difference of $2^{\beta(k)}$ or less terms that are powers of $2$. Therefore, from \cref{obs:num_blocks}, $bl(k)\leq 2^{\beta(k)}-1 < bl(k)+1-1 = bl(k)$. This is a contradiction.

\end{proof}
From \Cref{alpha_beta_simp_conn}, it is also easy to see the following corollary now. 
\begin{corollary}
	For every $k\in \mathbb{N}$, $\log (bl(k)+1) \leq \alpha(k)$     
\end{corollary}

\section{Proof of Upper Bound}\label{sec:upper}

We first establish a framework for analyzing the size of ideals generated by 
families of sets. Our approach leverages the inclusion-exclusion principle 
and introduces two key operations—splitting and lifting—that will be central 
to our construction.

\subsection{Technical Preliminaries}

A key insight for our analysis is understanding how the inclusion-exclusion 
principle applies to unions of power sets. This is captured in the following 
observation:

\begin{observation}\label{inclexclpowset}
	For finite sets $V_1,V_2\cdots V_q$
	$$
	|\bigcup_{i=1}^{q} 2^{V_i} | = \sum_{j=1}^{q}(-1)^{j+1} ( \sum_{1 \leq 
		i_1 < \cdots < i_j \leq q} 2^{|V_{i_1} \cap \cdots \cap V_{i_j}|} )
	$$      
\end{observation}

\begin{proof}
	From the inclusion-exclusion principle,
	$$
	|\bigcup_{i=1}^{q} 2^{V_i} | = \sum_{j=1}^{q}(-1)^{j+1} ( \sum_{1 \leq 
		i_1 < \cdots < i_j \leq q} |2^{V_{i_1}} \cap \cdots \cap 2^{V_{i_j}}| )
	$$
	$$
	= \sum_{j=1}^{q}(-1)^{j+1} ( \sum_{1 \leq i_1 < \cdots < i_j \leq q} 
	|2^{V_{i_1} \cap \cdots \cap V_{i_j}}| )
	= \sum_{j=1}^{q}(-1)^{j+1} ( \sum_{1 \leq i_1 < \cdots < i_j \leq q} 
	2^{|V_{i_1} \cap \cdots \cap V_{i_j}|} )
	$$     
\end{proof}

\subsection{Fundamental Operations: Splitting and Lifting}

We now introduce two fundamental operations that will serve as building 
blocks for our upper-bound construction. These operations allow us to 
construct ideals with specific cardinalities efficiently. 

\begin{lemma}[\textit{Splitting lemma}]\label{split} 
	For $m,k \in \mathbb{N}$, $\alpha(m+k)\leq \alpha(m)+\alpha(k+1)$
\end{lemma}

\begin{proof}
	Let $\mathcal{S},\mathcal{T}$ be family of sets such that $S\cap 
	T=\emptyset$ for all $S\in \mathcal{S}$ and $T\in \mathcal{T}$, where 
	$$|\mathcal{S}|=\alpha(m) \And |\mathbf{ID}(\mathcal{S})|=m$$
	$$|\mathcal{T}|=\alpha(k+1) \And |\mathbf{ID}(\mathcal{T})|=k+1$$
	
	Observe that, by construction $\mathbf{ID}(\mathcal{S}) \cap 
	\mathbf{ID}(\mathcal{T})=\{\emptyset\}$. Therefore, 
	$$|\mathbf{ID}(\mathcal{S} \cup \mathcal{T})| = |\mathbf{ID}(\mathcal{S})|+ 
	|\mathbf{ID}(\mathcal{T})|-1=m+k$$
	
	It now follows that, 
	$$\alpha(m+k) \leq |\mathcal{S} \cup \mathcal{T}|=|\mathcal{S}|+  
	|\mathcal{T}| = \alpha(m)+\alpha(k+1) $$
\end{proof}

The splitting lemma essentially tells us that to construct an ideal of size 
$m+k$, we can combine ideals of sizes $m$ and $k+1$ that share only the empty 
set. This allows us to decompose the problem of constructing larger ideals 
into constructing smaller ones.

Our second fundamental operation is the lifting lemma, which provides an 
efficient way to construct ideals whose cardinality is a power of 2 
multiplied by a given number:

\begin{lemma}[\textit{Lifting lemma}]\label{lift}
	For every $t,k \in \mathbb{N}$, $ \alpha(2^t\cdot k) \leq \alpha(k)$    
\end{lemma}

\begin{proof}
	Let $\mathcal{S}=\{S_1,S_2\cdots S_{\alpha(k)}\}$ be a family of sets such 
	that $ |\mathbf{ID}(\mathcal{S})|=k$. Let the set $X=\{x_1,x_2\cdots x_t\}$ 
	be such that $ X \cap S_i = \emptyset$ for all $S_i\in \mathcal{S}$. We 
	define $S_i'=S_i \sqcup X$ for all $i \in [\alpha(k)]$ and 
	$\mathcal{S'}=\{S_1',S_2'\cdots S_{\alpha(k)}'\}$. By 
	\cref{inclexclpowset},
	
	$$|\mathbf{ID}(\mathcal{S}')|=|\bigcup_{i=1}^{\alpha(k)} 2^{S_i'} | = 
	\sum_{j=1}^{\alpha(k)}(-1)^{j+1} ( \sum_{1 \leq i_1 < \cdots < i_j \leq 
		\alpha(k)} 2^{|S_{i_1}'\cap \cdots \cap S_{i_j}'|} )$$
	
	By construction,
	$$=\sum_{j=1}^{\alpha(k)}(-1)^{j+1} ( \sum_{1 \leq i_1 < \cdots < i_j 
		\leq \alpha(k)} 2^{|(S_{i_1}\cap \cdots \cap S_{i_j})\sqcup X|} )$$
	$$=2^{|X|}\sum_{j=1}^{\alpha(k)}(-1)^{j+1} ( \sum_{1 \leq i_1 < \cdots < 
		i_j \leq \alpha(k)} 2^{|S_{i_1}\cap \cdots \cap S_{i_j}|} )$$
	
	By \cref{inclexclpowset},
	$$2^{|X|}|\bigcup_{i=1}^{\alpha(k)} 2^{S_i} |=2^t|\mathbf{ID}(\mathcal{S})|
	=2^t\cdot k$$
	
	Therefore, $\alpha(2^t \cdot k)\leq |\mathcal{S}'| = |\mathcal{S}|=\alpha(k)$
\end{proof}

The lifting lemma provides a powerful tool: it shows that we can construct an 
ideal of size $2^t \cdot k$ using the same number of generators as an ideal 
of size $k$. Intuitively, this is achieved by adding $t$ new elements to each 
generator set in a way that preserves the relative structure of the original 
ideal.

\subsection{Simple Upper Bound Based on Block Count}

Using the operations we've developed, we can now establish a simple 
relationship between $\alpha(k)$ and the block count of $k$:

\begin{lemma}\label{simple_upper_bound}
	For every $k\in \mathbb{N}$, $\alpha(k)\leq {bl(k)}+1$    
\end{lemma}

\begin{proof}
	We prove this by induction on the block count of $k$. We first give a 
	construction for the base case, that is, for any $k$ with $bl(k)=1$, say 
	$k=\mathbf{1}_{q_1} \mathbf{0}_{l_1}$, where $q_1\geq 1$ and $l_1\geq 0$. 
	Consider the sets $S_1=[q_1+l_1-1]_1$ and $S_2=[q_1-1]_2\sqcup [l_1]_1$. 
	Note that $S_1 \bigcap S_2 = [l_1]_1$. We now observe that,
	$$ |\mathbf{ID}(\{S_1,S_2\})| = 2^{|S_1|}+2^{|S_2|}-2^{|S_1\bigcap S_2|} = 
	2^{q_1+l_1}-2^{l_1} = \mathbf{1}_{q_1} \mathbf{0}_{l_1}$$
	
	Thus for $k$, with $bl(k)=1$, $\alpha(k) \leq 2$. We now consider any number 
	$k$ with $bl(k)=b \geq 2$, say $1_{q_b}0_{l_b}\cdots1_{q_2}0_{l_2}1_{q_1}
	0_{l_1}$, where $q_i> 0$ and $l_j> 0$ for all $i\in [b]$ and $j \in [2,b]$. 
	From the \textit{Lifting lemma}, 
	$$\alpha(k) = \alpha (1_{q_b}0_{l_b}\cdots1_{q_2}0_{l_2}1_{q_1}0_{l_1}) 
	\leq \alpha (1_{q_b}0_{l_b}\cdots1_{q_2}0_{l_2}1_{q_1}) = \alpha 
	(1_{q_b}0_{l_b}\cdots1_{q_2}0_{l_2+q_1}+1_{q_1})$$
	
	From the \textit{Splitting lemma}, 
	$$ \leq \alpha (1_{q_b}0_{l_b}\cdots1_{q_2}0_{l_2+q_1})+\alpha(1_{q_1}+1) = 
	\alpha (1_{q_b}0_{l_b}\cdots1_{q_2}0_{l_2+q_1})+\alpha(2^{q_1})$$
	
	From the induction assumption, 
	$$ \leq b+\alpha(2^{q_1}) = b+1=bl(k)+1$$
\end{proof}

This lemma establishes that $\alpha(k)$ grows no faster than the block count 
of $k$ plus one. While this already gives us a non-trivial upper bound, we 
will develop tighter bounds in the next section.

\subsection{Tighter Upper Bound Construction}

We now present our main technical result, which establishes a much tighter 
bound on $\alpha(k)$. The key insight is to construct specialized sets for 
numbers of a particular form and then extend these constructions to all 
natural numbers.

\begin{theorem}
	\label{thm:basecase_recursion}
	For $m$ of the form $2^{3q^2}+\beta$ where $\beta < 2^{q^2}$, $\alpha(m) 
	\leq (q+1) \lceil \log{q} \rceil+4q+6$
\end{theorem}

The proof of this theorem is technical and will be presented in 
\Cref{subsec:proofthm11}. First, we'll show how this theorem helps us 
establish our main result, \Cref{thm:sqrt_bnd}.

\begin{observation}\label{sim_ind_obs}
	For every $k\in \mathbb{N}$, there exists a $q$ such that
	$k=2^{3q^2}+\gamma \cdot 2^{q^2}+\beta < 2^{3(q+1)^2}$ where $\gamma = 
	\lfloor \dfrac{k-2^{3q^2}}{2^{q^2}} \rfloor$ and $0 \leq \beta < 2^{q^2}$
\end{observation}  

\begin{proof}
	For every number $k$, there exists a $q$ such that $2^{3q^2}\leq k < 
	2^{3(q+1)^2}$. By dividing $k-2^{3q^2}$ by $2^{q^2}$, it follows that 
	$\gamma = \lfloor \dfrac{k-2^{3q^2}}{2^{q^2}} \rfloor$ is the quotient and 
	$0 \leq \beta < 2^{q^2}$ is the remainder. 
\end{proof}

This observation shows that any natural number can be decomposed into the 
form required by \Cref{thm:basecase_recursion}, plus an additional term. 
We now show how to handle this additional term:

\begin{lemma} \label{break_induction}
	$\alpha(2^{3q^2}+\gamma \cdot 2^{q^2}+\beta)\leq \alpha(2^{3q^2}+\beta)+ 
	\alpha(\gamma)+1$    
\end{lemma}

\begin{proof}
	From the Splitting lemma,   
	$$\alpha(2^{3q^2}+\gamma \cdot 2^{q^2}+\beta)\leq  \alpha(2^{3q^2}+\beta)+
	\alpha(\gamma\cdot 2^{q^2}+1)$$
	$$\leq \alpha(2^{3q^2}+\beta)+\alpha(\gamma\cdot 2^{q^2}) + \alpha(2)= 
	\alpha(2^{3q^2}+\beta)+\alpha(\gamma\cdot 2^{q^2}) + 1$$
	
	From the Lifting lemma, 
	$$=\alpha(2^{3q^2}+\beta)+ \alpha(\gamma)+1$$
\end{proof} 

With these pieces in place, we are now ready to prove our main result, which 
provides an $O(\sqrt{\log k}\log\log k)$ upper bound on $\alpha(k)$.

\begin{proof}[Proof of \Cref{thm:sqrt_bnd}]
	From \Cref{sim_ind_obs}, for every $k$, there exists a $q$ such that
	$k=2^{3q^2}+\gamma \cdot 2^{q^2}+\beta < 2^{3(q+1)^2}$ where $\gamma = 
	\lfloor \dfrac{k-2^{3q^2}}{2^{q^2}} \rfloor$ and $0 \leq \beta < 2^{q^2}$. 
	
	When $3 \le k < 20$, it is obvious that $\alpha(k) \le k < 20\sqrt{\log{k}}
	\log\log{k}$. So, we can assume that $20 \le k$ and therefore $q \ge 1$.
\begin{observation}\label{calc_der}
If $\log{3} \leq \log{k} \leq 
	30000$, then $\lceil 0.5\log{k}+1 \rceil < 20\sqrt{\log{k}}
	\log\log{k}$.   
\end{observation}    
We prove \cref{calc_der} in \cref{proof_of_calc_der}. 

From our simple upper bound in \Cref{simple_upper_bound} along with \cref{calc_der}, we get that for all $q<100$,
	$$\alpha(k) \leq \lceil 0.5\log{k} \rceil+1 < 20\sqrt{\log{k}}\log\log{k}$$
	
	We now inductively prove \Cref{thm:sqrt_bnd} for $q \geq 100$, inducting on 
	$q$. 
	$$\alpha(k) = \alpha(2^{3q^2}+\gamma \cdot 2^{q^2}+\beta)$$
	
	When ${\gamma}=0$, by \Cref{thm:basecase_recursion}, $\alpha(k) \leq 
	20\sqrt{\log{k}}\log\log{k}$. For ${\gamma}>0$, from \Cref{break_induction}, 
	$$ \leq \alpha(2^{3q^2}+\beta)+ \alpha(\gamma)+1$$
	
	From \Cref{thm:basecase_recursion}, 
	$$ \leq (q+1) \lceil \log{q} \rceil+4q+7 + \alpha(\gamma)$$
	
	When ${\gamma}\leq 2$ and $\log{k} > 30000$, as $\alpha(1)=\alpha(2)=1$, 
	$\alpha(k) \leq 20\sqrt{\log{k}}\log\log{k}$. Therefore, for ${\gamma} \geq 
	3$, by induction assumption,
	$$ \alpha(k) \leq (q+1) \lceil \log{q} \rceil+4q+7 + 20\sqrt{\log \gamma}
	\log\log{\gamma} $$
	
	As $\sqrt{\log{\gamma}} \leq \sqrt{\log{k} -q^2} \leq \sqrt{3(q+1)^2-q^2} 
	\leq \sqrt{2.1q^2}$ for $q\geq 100$,  
	$$ \leq (q+1) \lceil \log{q} \rceil+4q+7 + 20\sqrt{2.1}q\log\log{k}$$
	
	As $(q+1) \lceil \log{q} \rceil+4q+7 \leq 2q\log{q}$ for $q \geq 100$,
	$$ \leq 2q\log{q} + 20\sqrt{2.1}q\log\log{k} \leq q\log\log{k} + 
	20\sqrt{2.1}q\log\log{k}$$
	$$ =(20\sqrt{2.1}+1)q\log\log{k} \leq 20\sqrt{3}q\log\log{k} \leq 
	20\sqrt{\log k}\log\log{k}$$
\end{proof}

\subsection{Proof of \texorpdfstring{\Cref{thm:basecase_recursion}}
	{Theorem \ref{thm:basecase_recursion}}}
\label{subsec:proofthm11}

We now present the most technical part of our proof: constructing an ideal 
with specific properties for numbers of the form $2^{3q^2}+\beta$ where 
$\beta < 2^{q^2}$. The key idea is to carefully design a collection of sets 
based on the binary representation of $\beta$.

Let $m=2^{3q^2}+\beta$ where $\beta < 2^{q^2}$. For $i,j \in [0,q-1]$, we 
define $F_{ij}$ in the following way: 
\begin{itemize}
	\item If the $(jq+i+1)$th least significant bit of $\beta$ is $1$, then 
	fix $F_{ij}=\emptyset$. \\ Let $\mathcal{F}_0$ be the family of all such 
	sets.
	\item If the $(jq+i+1)$th least significant bit of $\beta$ is $0$, then 
	fix $F_{ij}=[i]_{jq+i}$. \\ Let $\mathcal{F}_1$ be the family of all such 
	sets.
\end{itemize}
We note that the least significant bit of $\beta$ is indexed to be the $1$st bit. On the other hand the indices $i,j$ start from $0$. From this construction, several important properties immediately follow:

\begin{remark}\label{fij_disj2}
	For all $i,j \in [0,q-1]$, $F_{ij} \cap [q^2]=\emptyset$    
\end{remark}

\begin{remark}\label{fij_disj}
	For any $F_{ij},F_{i'j'} \in \mathcal{F}_1$, $F_{ij}\bigcap F_{i'j'} \neq 
	\emptyset$ if and only if $i=i'$ and $j=j'$.   
\end{remark}

\begin{remark}\label{rem:beta1}
	$\beta+ \sum\limits_{F_{ij}\in \mathcal{F}_1} 2^{jq+|F_{ij}|} =\beta+
	\sum\limits_{F_{ij}\in \mathcal{F}_1} 2^{jq+i} = \mathbf{1}_{{q^2}} = 
	2^{q^2}-1$
\end{remark}

\begin{remark}\label{rem:beta2}
	$\sum\limits_{F_{ij}\in \mathcal{F}_0} 2^{jq+|F_{ij}|} =\sum\limits_{F_{ij}
		\in \mathcal{F}_0} 2^{jq} $
\end{remark}

These observations lead to the following key relationship:

\begin{corollary}\label{fij_beta}
	$2^{q^2}-\sum\limits_{i,j \in [0,q-1]} 2^{jq+|F_{ij}|}= \beta +1 - 
	\sum\limits_{F_{ij}\in \mathcal{F}_0} 2^{jq}$     
\end{corollary}

Using these $F_{ij}$ sets, we construct two families of sets that will form 
the basis of our ideal. For all $i,j \in [0,q-1]$, let $S_i'=[q^2], 
T_j'=[jq]$ and let  
$$S_i=(\bigcup_{j=0}^{q-1}F_{ij}) \bigsqcup S_i' \And T_j=
(\bigcup_{i=0}^{q-1}F_{ij}) \bigsqcup T_j'$$

We define $\mathcal{S}=\{S_0,\cdots S_{q-1}\}, \mathcal{S}'=\{S_0',\cdots 
S_{q-1}'\}, \mathcal{T}=\{T_0,\cdots T_{q-1} \}$ and $\mathcal{T}'=
\{T_0',\cdots T_{q-1}' \}$

Our strategy is to compute the cardinality of the ideal generated by 
$\mathcal{S}\cup \mathcal{T}$. We want to show that $\alpha(|\mathbf{ID} 
(\mathcal{S}\cup \mathcal{T})|) \leq 2q$, and since $q = O(\sqrt{\log k})$, 
this would give us the desired upper bound if $|\mathbf{ID} (\mathcal{S}\cup 
\mathcal{T})| = \beta = m-2^{3q^2}$. However, this equality doesn't hold 
exactly, but we'll show that the difference has a small $\alpha$-value.

Computing $|\mathbf{ID}(\mathcal{S}\cup \mathcal{T})|$ directly is complex, 
so we first relate it to $|\mathbf{ID}(\mathcal{S}'\cup \mathcal{T}')|$, 
which equals $2^{q^2}$ because all sets in $\mathcal{S}'$ and $\mathcal{T}'$ 
are subsets of $[q^2]$ and every $S_i' = [q^2]$.

\begin{observation}\label{obs_fij_simp}
	$S_i \cap T_j = [jq]\bigsqcup F_{ij}$
\end{observation}

\begin{proof}
	By \Cref{fij_disj2} and \Cref{fij_disj}, 
	$$S_i \cap T_j = ((\bigcup_{j'=0}^{q-1}F_{ij}) \sqcup [q^2]) \cap 
	((\bigcup_{i'=0}^{q-1}F_{ij}) \sqcup [jq])= [jq]\bigsqcup F_{ij} $$
\end{proof}

\begin{observation}\label{bastwosetsyseq}
	For $p < i$, $2^{T_i} \cap 2^{T_p}= 2^{[pq]}= 2^{T'_i} \cap 2^{T'_p}$ and 
	$2^{S_i} \cap 2^{S_p}= 2^{[q^2]}=2^{S'_i} \cap 2^{S'_p}$
\end{observation} 

\begin{proof}
	By \Cref{fij_disj2} and \Cref{fij_disj}, 
	$$2^{T_i} \cap 2^{T_p} = 2^{T_i \cap T_p}= 2^{T_i' \cap T_p'}= 2^{[pq]} = 
	2^{T_i'} \cap 2^{T_p'} $$
	$$2^{S_i} \cap 2^{S_p} = 2^{S_i \cap S_p}=2^{S_i' \cap S_p'}= 2^{[q^2]} = 
	2^{S_i'} \cap 2^{S_p'} $$
\end{proof}

These observations allow us to relate terms in the inclusion-exclusion 
expansions of $|\mathbf{ID}(\mathcal{S}\cup \mathcal{T})|$ and 
$|\mathbf{ID}(\mathcal{S}'\cup \mathcal{T}')|$. We define:
$$(A_1,A_2\cdots A_{2q})=(2^{S_0},2^{S_1}\cdots 2^{S_{q-1}},2^{T_0},
2^{T_1}\cdots 2^{T_{q-1}})$$ 
$$(A_1',A_2'\cdots A_{2q}')=(2^{S_0'},2^{S_1'}\cdots 2^{S_{q-1}'},
2^{T_0'},2^{T_1'}\cdots 2^{T_{q-1}'})$$
We note that the indices of $A$ and $A'$ start at $1$ while indices of $S,S',T,T'$ start at $0$.

\begin{observation}\label{twosetsyseq}
	For any $\ell \geq 3$, let $I=\{i_1,i_2\cdots i_{\ell}\}$ such that $1 \leq 
	i_1 < \cdots < i_{\ell} \leq 2q$
	$$ |\bigcap\limits_{j \in I} A_j|= |\bigcap\limits_{j \in I} A_j'|$$
\end{observation}

\begin{proof}
	As $\ell \geq 3$, by pigeon hole principle, there exists two sets $A_{i_x}$ 
	and $A_{i_y}$ such that either $A_{i_x}, A_{i_y} \in \{2^{S_0},2^{S_1}\cdots 
	2^{S_{q-1}}\}$ or $A_{i_x}, A_{i_y} \in \{2^{T_0},2^{T_1}\cdots 
	2^{T_{q-1}}\}$. From \Cref{bastwosetsyseq}, $A_{i_x} \cap A_{i_y} \subset 
	2^{[q^2]}$. It follows that $\bigcap\limits_{j \in I} A_j \subset 2^{[q^2]}$. 
	Therefore, 
	$$ |\bigcap\limits_{j \in I} A_j|= |\bigcap\limits_{j \in I} (A_j\bigcap 
	2^{[q^2]})|  = |\bigcap\limits_{j \in I} A_j'|$$
\end{proof}

This allows us to compute the difference between the two ideals' 
cardinalities:

\begin{observation}\label{obs_fin_st}
	$$|\mathbf{ID}(\mathcal{S}\cup \mathcal{T})|-|\mathbf{ID}(\mathcal{S}'\cup 
	\mathcal{T}')|$$    
	$$=  (\sum_{i=0}^{q-1} |2^{S_i}| - \sum_{i=0}^{q-1} |2^{S_i'}|)+
	(\sum_{j=0}^{q-1} |2^{T_j}| - \sum_{j=0}^{q-1} |2^{T_j'}|)
	-(\sum_{i,j \in [0,q-1]} (|2^{S_i\bigcap T_j}|)-\sum_{i,j \in [0,q-1]} 
	(|2^{S_i'\bigcap T_j'}|))$$
\end{observation}

Through a series of algebraic manipulations (Proof in \cref{obsalgman}), we derive 
the exact cardinality of our constructed ideal:

\begin{lemma}\label{deflemcalc34}
	$|\mathbf{ID}(\mathcal{S}\cup \mathcal{T})| = \sum_{i=0}^{q-1} |2^{S_i}| 
	+ \sum_{j=0}^{q-1} |2^{T_j}| - \sum_{i,j \in [0,q-1]} 2^{jq}(2^{|F_{ij}|})
	+(q-1)(\sum_{j=0}^{q-1} 2^{jq} - 2^{q^2}) $
\end{lemma}
This leads to a bound on the $\alpha$-value of our constructed ideal (as $|\mathcal{S}\cup \mathcal{T}| = 2q$):

\begin{corollary}\label{cons:system1}
	$\alpha(\sum_{i=0}^{q-1} |2^{S_i}| + \sum_{j=0}^{q-1} |2^{T_j}| - 
	\sum_{i,j \in [0,q-1]} 2^{jq}(2^{|F_{ij}|})+(q-1)(\sum_{j=0}^{q-1} 2^{jq} 
	- 2^{q^2})) = \alpha(|\mathbf{ID}(\mathcal{S}\cup \mathcal{T})|) \leq 2q$
\end{corollary}

Our goal is to show that $\alpha(m - |\mathbf{ID}(\mathcal{S}\cup 
\mathcal{T})|) = O(q \log q)$, which would imply $\alpha(m) = 
O(q \log q) = O(\sqrt{\log k} \log \log k)$ by the splitting lemma and 
\Cref{cons:system1}. However, the expression for 
$|\mathbf{ID}(\mathcal{S}\cup \mathcal{T})|$ contains inconvenient terms 
like $\sum_{i=0}^{q-1} |2^{S_i}|$. We first eliminate these terms:

\begin{observation}\label{cons:system2}
	$\alpha(2^{3q^2-1}-\sum_{i=0}^{q-1} |2^{S_i}| - \sum_{j=0}^{q-1} 
	|2^{T_j}|+2^{q^2}) \leq 2q+3$
\end{observation}

\begin{proof}
	From \Cref{obs:num_blocks}, $2^{3q^2-1}-\sum_{i=0}^{q-1} |2^{S_i}| - 
	\sum_{j=0}^{q-1} |2^{T_j}|+2^{q^2}$ has at most $2q+2$ blocks. Therefore, 
	from our earlier results on block counts and $\alpha$ values,
	$\alpha(2^{3q^2-1}-\sum_{i=0}^{q-1} |2^{S_i}| - \sum_{j=0}^{q-1} 
	|2^{T_j}|+2^{q^2}) \leq 2q+3$  
\end{proof}

For clarity, we define two auxiliary values:
$t_1=2^{3q^2-1}+2^{q^2}-1 - \sum_{i,j \in [0,q-1]} 2^{jq}(2^{|F_{ij}|})
+(q-1)(\sum_{j=0}^{q-1} 2^{jq} - 2^{q^2})$

$t_2=2^{3q^2-1}+ \sum_{F_{ij}\in \mathcal{F}_0} 2^{jq}- (q-1)
(\sum_{j=0}^{q-1} 2^{jq} - 2^{q^2})$

Here, $t_1$ is the sum of $|\mathbf{ID}(\mathcal{S}\cup \mathcal{T})|$ and 
the expression from \Cref{cons:system2}. We now show that $t_2 = m - t_1$:

\begin{observation}\label{sumt1t2}
	$m=t_1+t_2$
\end{observation}

\begin{proof}
	By definition,
	$t_1+t_2= 2^{3q^2}+2^{q^2}-1-\sum\limits_{i,j \in [0,q-1]} 2^{jq+|F_{ij}|}
	+\sum\limits_{F_{ij}\in \mathcal{F}_0} 2^{jq}$
	
	From \Cref{fij_beta}, 
	$=2^{3q^2}+\beta = m$    
\end{proof}

\begin{observation}\label{mid_alpha_sets}
	$\alpha(t_1) \leq 4q+3 $   
\end{observation}

\begin{proof}
	This follows from the splitting lemma along with \Cref{cons:system1} and 
	\Cref{cons:system2}.
\end{proof}

By the splitting lemma and \Cref{sumt1t2}, it is now sufficient to prove that 
$\alpha(t_2+1)$ is small. We'll show that $t_2$ can be written in a special 
form that allows us to bound its $\alpha$-value efficiently:

\begin{observation}\label{pre_final_alpha_count}
	There exists $-1 \leq a_j \leq q-1$ for $j \in [0,q-1]$ and $a_q=-(q-1)$, 
	for which 
	$t_2 = 2^{3q^2-1}-\sum_{j=0}^{q} a_j2^{jq}$
\end{observation}

\begin{proof}
	$t_2= 2^{3q^2-1}+\sum_{F_{ij}\in \mathcal{F}_0} 2^{jq} - (q-1)
	(\sum_{j=0}^{q-1} 2^{jq} - 2^{q^2})$
	
	For a given $j$, let $0\leq b_j\leq q$ be the number of $F_{ij} \in 
	\mathcal{F}_0$
	$= 2^{3q^2-1} + (q-1)2^{q^2} - \sum_{j=0}^{q-1} (q-1-b_j)2^{jq}$
	
	Therefore, for some $-1 \leq a_j \leq q-1$ for $j \in [0,q-1]$ and 
	$a_q=-(q-1)$
	$=2^{3q^2-1} - \sum_{j=0}^{q} a_j2^{jq} $
\end{proof}

Finally, we show that numbers with this special structure have a small 
$\alpha$-value:

\begin{lemma}\label{final_alpha_count}
	For any $|a_j| \leq q-1$ for $j \in [0,q]$, 
	$\alpha(2^{3q^2-1}-\sum_{j=0}^{q} a_j2^{jq} + 1) \leq (q+1)\lceil 
	\log{q} \rceil+3$
\end{lemma}

\begin{proof}
	The binary representation of $|a_j|$ has at most $\lceil \log{q} \rceil$ 
	non-zero bits. Therefore, $|a_j|2^{jq}$ has at most $\lceil \log{q} \rceil$ 
	non-zero bits in its binary representation. This means $\sum_{j=0}^{q} 
	a_j2^{jq}$ can be written as the sum or difference of $(q+1)\lceil \log{q} 
	\rceil$ powers of $2$. From our earlier results on block counts,  
	$2^{3q^2-1}-\sum_{j=0}^{q} a_j2^{jq} + 1$ has at most $(q+1)\lceil 
	\log{q} \rceil+2$ blocks. Therefore,
	$\alpha(2^{3q^2-1}-\sum_{j=0}^{q} a_j2^{jq} + 1) \leq (q+1)\lceil 
	\log{q} \rceil+3$
\end{proof}

\begin{corollary}\label{t2corol}
	$\alpha(t_2+1)\leq (q+1)\lceil \log{q} \rceil+3$    
\end{corollary}

Combining all these results, \Cref{thm:basecase_recursion} follows from the 
splitting lemma along with \Cref{t2corol}, \Cref{mid_alpha_sets}, and 
\Cref{sumt1t2}.

\section{Conclusion and Open Problems}\label{sec:conclusion}

We have established that for every $k \geq 3$, there exists a DNF or CNF 
with exactly $k$ satisfying assignments using at most $O(\sqrt{\log k}
\log\log k)$ terms or clauses. Our construction provides the first 
$o(\log k)$ upper bound for this problem, significantly improving previous 
bounds \cite{Meel15}. The constructed DNFs also have the desirable property of being 
monotone, which simplifies their structure and analysis. On the other hand, we also give a lower bound showing that there exist infinitely many $k \in \mathbb{N}$ 
requiring at least $\Omega(\log\log k)$ terms or clauses. However, there 
remains a gap between our upper and lower bounds that presents an 
interesting avenue for future research.

We conjecture that the value of $\beta(k)$ is polynomial in 
$\log(bl(k) + 1)$, which would provide a more precise characterization 
of the relationship between the number of terms needed and the block 
structure of the number of solutions. Resolving this conjecture would 
further deepen our understanding of the structural properties of boolean 
formulas with a specific number of satisfying assignments.

The connection we established between this problem and the theory of ideals 
in set systems may also lead to further applications in other areas of 
combinatorics and computational complexity theory. In particular, the 
construction techniques we developed might be useful in addressing related 
questions about the expressiveness and succinctness of different 
representations of boolean functions.
\section*{Acknowledgments}
This research was funded in part by the Natural Sciences and Engineering Research Council
of Canada (NSERC), funding reference number RGPIN-2024-05956.

\appendix
\appendix
\section{Deferred proofs}
\subsection{Proof of \texorpdfstring{\cref{calc_der}}{Observation 24}}\label{proof_of_calc_der}

We claim that for every integer $3 \le k \le 2^{30000}$,
\[
\Bigl\lceil \tfrac12\log_2 k + 1 \Bigr\rceil
\;<\;
20\,\sqrt{\log_2 k}\,
      \log_2\!\bigl(\log_2 k\bigr).
\]

\begin{proof}
Put $x:=\log_2 k$.  Then $x\in[\log_2 3,\,30000]$.  The inequality
\(\lceil y\rceil<y+1\) gives
\[
\Bigl\lceil \tfrac12 x + 1 \Bigr\rceil
\;<\;
\tfrac12 x + 2,
\]
so it suffices to prove
\[
f(x):=\tfrac12 x + 2 - 20\sqrt{x}\,\log_2 x<0
\quad\text{for }x\in[\log_2 3,\,30000].
\]

\medskip
It is easy to see that
\[
f'(x)=\frac12-\frac{10\log_2 x}{\sqrt{x}}
      -\frac{20\log_2 e}{\sqrt{x}}, 
\qquad
f''(x)
      =\frac{5\,\log_2 x}{x^{3/2}}
      \;>\;0
      \quad(x>1).
\]
Hence \(f\) is convex on the entire interval. We can also compute that
\[
f(\log_2 3)\approx -13.8\;<\;0,
\qquad
f(30000)\approx -3.66\times 10^{4}\;<\;0.
\]

As both the end points of the convex function are negative, we know that the function is also negative everywhere between them. So, \(f(x)<0\) for all
\(x\in[\log_2 3,30000]\).
Substituting \(x=\log_2 k\) yields the claimed inequality.
\end{proof}

\subsection{Proof of \texorpdfstring{\cref{deflemcalc34}}{Lemma 34}}\label{obsalgman}
Observe that $$ \sum_{i=0}^{i=q-1} |2^{S_i'}| = q\cdot 2^{q^2}, \sum_{j=0}^{j=q-1} |2^{T_j'}| = \sum_{j=0}^{j=q-1} 2^{jk}$$
Moreover, since all sets in $\mathcal{S}'$ and $\mathcal{T}'$ are subsets of $[q^2]$
$$|\mathbf{ID}(\mathcal{S}'\cup \mathcal{T}')|=2^{q^2}$$
$$ \sum_{i,j \in [0,q-1]} (|2^{S_i'\bigcap T_j'}|) = \sum_{i,j \in [0,q-1]} (|2^{[jk]}|) =q\cdot \sum_{i=0}^{i=q-1} 2^{ik} $$
Substituting these values in \Cref{obs_fin_st}, we get $|\mathbf{ID}(\mathcal{S}\cup \mathcal{T})|$
$$ = \sum_{i=0}^{i=q-1} |2^{S_i}| + \sum_{j=0}^{j=q-1} |2^{T_j}| - \sum_{i,j \in [0,q-1]} (|2^{S_i\bigcap T_j}|)
-q\cdot 2^{q^2} - \sum_{i=0}^{i=q-1} 2^{ik} + q\cdot \sum_{i=0}^{i=q-1} 2^{ik} + 2^{q^2}$$

$$ = \sum_{i=0}^{i=q-1} |2^{S_i}| + \sum_{j=0}^{j=q-1} |2^{T_j}| - \sum_{i,j \in [0,q-1]} (|2^{S_i\bigcap T_j}|) +(q-1)(\sum_{j=0}^{j=q-1} 2^{jk} - 2^{q^2})$$
From \Cref{obs_fij_simp},

$$ = \sum_{i=0}^{i=q-1} |2^{S_i}| + \sum_{j=0}^{j=q-1} |2^{T_j}| - \sum_{i,j \in [0,q-1]} 2^{jk}(2^{|F_{ij}|}) +(q-1)(\sum_{j=0}^{j=q-1} 2^{jk} - 2^{q^2})$$

\end{document}